\documentclass{birkjour}
\usepackage{amssymb,mathrsfs}
\usepackage{enumerate}
\usepackage{verbatim}
\usepackage{array}
\usepackage[dvipsnames]{xcolor}
\usepackage[all]{xy}
\usepackage[colorlinks=true, linkcolor=BrickRed, urlcolor=Blue,citecolor=OliveGreen,menucolor=]{hyperref}

\theoremstyle{plain}
\newtheorem{thm}{Theorem}

\newtheorem{prop}[thm]{Proposition}
\theoremstyle{definition}
\newtheorem{defn}[thm]{Definition}
\theoremstyle{remark}
\newtheorem{rem}[thm]{Remark}
\newtheorem{example}[thm]{Example}

\newcommand{\norm}[1]{\left\Vert#1\right\Vert}
\renewcommand{\sc}[2]{\langle #1|#2 \rangle}

\newcommand{\R}{\mathbb R}

\renewcommand{\H}{\mathcal{H}}
\newcommand{\g}{\mathfrak{g}}

\DeclareMathOperator{\Tr}{Tr}

\newcommand{\pb}{\{\cdot,\cdot\}}

\renewcommand{\to}{\rightarrow}

\newcommand{\oth}{\,\widehat{\otimes} \,}
\newcommand{\Ci}{\mathscr{C}^\infty}

\newcommand{\todot}[1]
  {\vspace{5 mm}\par \noindent \marginpar{\LARGE TG} \framebox{\begin
  {minipage}[c]{0.95 \textwidth} \tt\color{OliveGreen} #1
\end{minipage}}\vspace{5 mm}\par}
\newcommand{\todob}[1]
  {\vspace{5 mm}\par \noindent \marginpar{\LARGE ABT} \framebox{\begin
  {minipage}[c]{0.95 \textwidth} \tt\color{WildStrawberry} #1
\end{minipage}}\vspace{5 mm}\par}
\newcommand{\todop}[1]
  {\vspace{5 mm}\par \noindent \marginpar{\LARGE PR} \framebox{\begin
  {minipage}[c]{0.95 \textwidth} \tt\color{blue} #1
\end{minipage}}\vspace{5 mm}\par}

\begin{document}
\bibliographystyle{grz2}
\title[Poisson structures in infinite dimensions]{Poisson structures in the Banach setting: comparison of different approaches}
\author[T.~Goli\'nski]{Tomasz Goli\'nski}
\address{University of Bia\l ystok\\ Cio\l kowskiego 1M\\15-245 Bia\l ystok\\ Poland}
\email{tomaszg@math.uwb.edu.pl}
\author[P.~Rahangdale]{Praful Rahangdale}
\address{Universit\"at Paderborn\\
    Institut f\"ur Mathematik\\
    Warburger Str. 100\\
    33098 Paderborn\\ Germany}

    \email{praful@math.uni-paderborn.de}

\author[A.B.~Tumpach]{Alice Barbora Tumpach}
\address{
\begin{minipage}{0.5\textwidth}
UMR CNRS 8524\\
UFR de Math\'ematiques\\
Laboratoire Paul Painlev\'e\\
59 655 Villeneuve d'Ascq Cedex\\ France \\
\end{minipage}
\begin{minipage}{0.5\textwidth}
Institut CNRS Pauli\\ UMI  CNRS 2842\\ Oskar-Morgenstern-Platz 1 \\1090 Wien\\Austria \\
\end{minipage}
}
\email{alice-barbora.tumpach@univ-lille.fr}
\thanks{The first and third authors were supported by the 2020 National Science Centre, Poland / Fonds zur Förderung der wissenschaftlichen Forschung, Austria grant ``Banach Poisson--Lie groups and integrable systems'' 2020/01/Y/ST1/00123, I-5015N. The second author was funded by TRR 358 ``Integral Structures in Geometry and Representation Theory'' (SFB-TRR 358/1 2023 – 491392403).}

\begin{abstract}
In this paper we examine various approaches to the notion of Poisson manifold in the context of Banach manifolds. Existing definitions are presented and differences between them are explored and illustrated with examples.
\end{abstract}

\maketitle


\section{Introduction}

Poisson brackets are very useful tools in classical mechanics, where they are usually considered on finite-dimensional manifolds. However, the theory of integrable systems (e.g. Korteweg--de Vries hierarchy, non-linear Schr\"odinger, see \cite{fadeev,dubrovin1984poisson,grabowski94,kolev07} for more examples) uses a similar approach on function spaces or diffeomorphism groups, which are infinite-dimensional. This point of view is usually only formal. It is a long-standing issue to formalize this approach, and the study of Poisson brackets on Banach manifolds can be seen as a first step in this direction, even though the spaces needed for the study of integrable systems are usually not modeled on Banach spaces. Some integrable systems, though, fit in the Banach framework, see e.g. \cite{OR,Oind,DO-L2,GO-grass,grabowski18,GT-momentum}.

In the context of Banach manifolds $M$, symplectic geometry has two distinct flavors: there are strong symplectic manifolds and weak symplectic manifolds. In both cases, the symplectic form defines a ``flat'' map $\flat$, which maps the tangent bundle $TM$ into the cotangent bundle $T^*M$. It is always injective as a consequence of the non-degeneracy of the symplectic form, but in the infinite-dimensional case, it does not need to be surjective as one cannot employ dimension counting argument.  

The typical example of symplectic manifolds is given by cotangent bundles. The canonical symplectic form on the cotangent space of a Banach manifold turns out to be strong if and only if the manifold is modeled on reflexive Banach spaces, see \cite[\S 1, Theorem 3]{marsden-chernoff}.


In the case of strong symplectic manifolds the flat map $\flat$ is invertible and it's inverse can be interpreted as a Poisson tensor. However, even in finite dimensions, not every Poisson tensor can be obtained in this way, as it may not be invertible, for example, the Kirillov--Kostant--Souriau Lie--Poisson tensor on the dual of a Lie algebra vanishes at the origin. In the weak case, when $\flat$ is not surjective, one can consider its inverse on its image. In general, this image does not need to be closed, and thus one may need to introduce a different Banach structure on it. The result is a kind of Poisson tensor, which is defined only on a subset of the cotangent bundle given by the image of the flat map $\flat$. In consequence, it leads to Poisson brackets, which are not defined on the whole algebra $\Ci(M)$ of smooth functions on the Banach manifold $M$, but rather on a subalgebra $\mathcal{A}$ consisting of functions whose differentials belong to the image of the flat map $\flat$ (see Theorem 48.8 in \cite{michor} and Proposition 3.11. in \cite{tumpach-bruhat}). In the literature, they were called weak Poisson brackets, sub-Poisson brackets, partial Poisson brackets, or generalized Poisson brackets.  This situation is similar to the concept of sub-Riemannian geometry, where the Riemann tensor is only defined on a subbundle of the tangent bundle. Moreover, for some of the definitions of Poisson structures, Hamiltonian vector fields may not exist for all the functions for which the bracket is defined.

Let us note that we confine ourselves to the context of Banach manifolds. Some results were obtained in a more general context, such as locally convex vector spaces \cite{glockner2009,neeb14,glockner22} or convenient spaces \cite{pelletier19,pelletier24}, vertex algebra, and $\lambda$-brackets \cite{valeri15,valeri16,valeri16a}. We include the approaches from \cite{neeb14, pelletier19}, but only in the Banach context, in order to be able to compare these with other approaches in the Banach case.

\section{Notation}

For any Banach space $E$, $E^*$ will denote the continuous dual of $E$ endowed with its natural Banach structure. Let $E_1,E_2,\dots$ be Banach spaces. We will denote by $L(E_1;E_2)$ the space of all bounded linear maps $E_1\to E_2$. By $B(E_1,E_2;E_3)$ we will denote the space of all bounded bilinear maps $E_1\times E_2\to E_3$. Note that in the Banach context it is isometrically isomorphic to $L(E_1;L(E_2;E_3))$.

In this paper, $M$ will denote a smooth Banach manifold (unless otherwise specified in some particular instances). The algebra of all smooth functions on $M$ will be denoted by $\Ci(M)$. The cotangent bundle of $M$ will be denoted by $T^*M$.  The natural duality pairing between $T^*M$ and $TM$ will be denoted by $\langle\cdot, \cdot\rangle$. 
The space of smooth sections of a vector bundle $V$ over $M$ will be denoted by $\Gamma(V)$. In particular, $\Gamma(T^*M)$ is the space of one-forms on $M$ and $\Gamma(TM)$ the space of vector fields on $M$. The Lie derivative along a vector field $X\in \Gamma(TM)$ will be denoted by $\mathcal{L}_{X}$.

\begin{defn}\label{def:Jacobi}
    A bilinear antisymmetric operation $\{\cdot, \cdot\}$ on a commutative associative algebra $\mathcal A$ satisfies the \textbf{Jacobi identity} iff
    \begin{equation}\label{Jacobi} 
\{a,\{b,c\}\} + \{b,\{c,a\}\} + \{c,\{a,b\}\} = 0,
\end{equation}
for all $a,b,c\in\mathcal A$.
\end{defn}

\begin{defn}\label{def:Leibniz}
    A bilinear antisymmetric operation $\{\cdot, \cdot\}$ on a commutative associative algebra $\mathcal A$ satisfies \textbf{Leibniz rule} iff
    \begin{equation}\label{Leibniz}
        \{ ab, c\} = \{ a ,c\}b + a\{ b,c\} ,
    \end{equation} 
    for all $a,b,c\in\mathcal A$.
\end{defn}
\begin{defn}\label{Poisson_algebra}
A \textbf{Poisson algebra} $\mathcal A$ is a commutative associative algebra endowed with a bilinear antisymmetric operation $\pb:\mathcal A\times \mathcal A\to \mathcal A$ such that
\begin{enumerate}
\item $\pb$ satisfies Jacobi identity,

\item $\pb$ satisfies Leibniz identity.

\end{enumerate}

\end{defn}




\section{Poisson structures defined for the whole algebra of smooth functions on a Banach manifold}\label{sec:strong}

We start the review of various approaches to Poisson geometry by the definitions of \textbf{Poisson brackets} on the whole algebra $\Ci(M)$ of smooth functions on a Banach manifold $M$. The definition of Poisson bracket on Banach manifolds was first proposed in the paper \cite{OR}. Let us state it here.

\begin{defn}[Odzijewicz--Ratiu \cite{OR}, Definition 2.1]\label{def:OR}
A Banach Poisson manifold is a pair $(M, \pb)$ consisting of a smooth Banach manifold and a bilinear operation $\pb$ satisfying the following conditions:
\begin{enumerate}
    \item[(i)] $\left(\Ci(M), \{\cdot, \cdot\}\right)$ is a Lie algebra;
    \item[(ii)] $\{\cdot,\cdot\}$ satisfies the Leibniz identity on each factor;
    \item[(iii)] the vector bundle map $\sharp: T^*M\to T^{**}M$ covering the identity defined by \[\sharp_m(dh(m))(dg(m)) = \{g, h\}(m), m \in M\] for any locally defined functions $g$ and $h$, satisfies $\sharp(T^*M)\subset TM$.
\end{enumerate}

\end{defn}

\begin{rem}
Conditions (i) and (ii) mean that $\left(\Ci(M), \{\cdot, \cdot\}\right)$ is a Poisson algebra (see Definition~\ref{Poisson_algebra}) and condition (iii) implies the existence of Hamiltonian vector fields.
\end{rem}

Note that Definition~\ref{def:OR} implicitly assumes the existence of a Poisson tensor. However, this is not automatic, even on a Hilbert space (see \cite{BGT}). If the Poisson tensor does not exist, it can happen that the value of the Poisson bracket $\{f,g\}$ of two functions $f$ and $g$ at a given point $m\in M$ depends on higher derivatives than $df(m)$ and $dg(m)$. This type of Poisson bracket has gone by the (diversity-friendly) name of queer Poisson brackets in reference to queer tangent vectors introduced in \cite{michor}, which were used for the construction of explicit examples of queer Poisson brackets. In this case, the map $\sharp$ cannot be defined. More generally, a Poisson bracket may not be localizable, i.e. may depend on the global values of functions, and in this case it is not possible to compute a Poisson bracket of functions which are only locally defined. 

Let us clarify these points. First, let us give the definition of a localizable Poisson bracket, following \cite{BGT}.
\begin{defn}\label{def:localizable}
A Poisson bracket $\pb$ on $M$ is called {\bf localizable} if it has a localization, that is, a family consisting of a Poisson bracket $\pb_U$ on every open subset $U\subseteq M$, which satisfy $\pb_M=\pb$ and are compatible with restrictions, i.e., if $U\subseteq V$ and $f,g\in \Ci(V)$ then $\{f,g\}_V\vert_U=\{f\vert_U,g\vert_U\}_U$. 
\end{defn}

In the finite-dimensional setting, all Poisson brackets are localizable. The usual proof utilizes the notion of bump functions, which are one of the basic tools in finite-dimensional differential geometry, but may not exist in the infinite-dimensional setting. A summary of results related to this topic can be found in \cite[Chapter III]{michor}.
\begin{defn}
By a \emph{bump function} on a Banach space $E$ we will understand a smooth real or complex-valued non-zero function whose support defined as
\[
  \operatorname{supp}f := \overline{ \{ x\in E | f(x)\neq 0 \} }.
\]
is a bounded subset of $E$.
\end{defn}
Choosing appropriate chart, it is straightforward to show that if there exists at least one bump function on the modeling space of a Banach manifold, then for any point $x_0\in M$ and any open neighborhood $U$ of $x_0$, there exists a bump function $\chi$ such that $\chi(x_0)=1$ and $\operatorname{supp}\chi\subset U$.

On many Banach spaces, bump functions do not exist. It is the case, for example, for $\ell^1$ spaces; see \cite{bonic-frampton1966} for the proof and other related results. In the finite-dimensional case, the existence of bump functions follows from the existence of an equivalent norm which is smooth away from the origin. However, in an arbitrary Banach space, there might not exist such a norm, see \cite{restrepo1964}. The existence of such a norm guarantees the existence of smooth bump functions. The construction is similar to the one in finite dimensions, for example they can be obtained by manipulating the function $e^{-\frac{1}{{\norm x}^2}}$, see \cite{ratiu-mta} for the example of explicit construction. Note that in \cite{restrepo1964} it was shown that a separable Banach space admits an equivalent norm of class $C^1$ if and only if its dual is separable. 
In particular, Banach spaces $\ell^1$ or $C([a,b])$ do not have equivalent $C^1$ norms.

\begin{prop}
The existence of bump functions on the modeling spaces of the manifold $M$ implies localizability of Poisson brackets (or more generally of derivations) on $M$.
\end{prop}
\begin{proof}
The first observation is that constant functions are Casimirs. Indeed, it is a consequence of the Leibniz rule since for the constant function $1$ and any function $f$, one has
\[ \{f, 1\}(x) = \{f, 1\cdot 1\}(x) = 1(x)\{f, 1\}(x) + 1(x)\{f, 1\}(x) = 2\{f,1\}(x) = 0.\]
Consider a smooth bump function $\chi$ such that $\chi(x_0)=1$. From the previous observation, it follows that for any functions $f,g$, we have
\[ \{ f, g \}(x) = \{ f - f(x_0), g - g(x_0) \}(x).\]
Thus without loss of generality, we can assume that $f(x_0)=g(x_0)=0$. Now
\[ \{ f, g\chi\}(x_0) = \{f,g\}(x_0)\chi(x_0) + \{ f,\chi\}(x_0)g(x_0) = \{f,g\}(x_0). \]
In effect, $\{f,g\}(x_0)$ depends only on the values of $f,g$ in an arbitrary neighborhood of $x_0$. It allows defining a Poisson bracket $\pb_U$ in a unique way by extending smooth functions on $U$ to smooth functions on the whole manifold $M$.
\end{proof}

To our knowledge the problem of the existence of non-localizable Poisson brackets on Banach manifolds is still open, see a related discussion in \cite{pelletier}. However, as mentioned above, even assuming that the Poisson bracket is localizable, it does not follow that $\{f, g\}(m)$ depends only on $df(m)$, $dg(m)$, which means that the sharp map $\sharp$ in Definition~\ref{def:OR} might still be ill defined. If the value of the Poisson bracket $\{f, g\}(m)$ depends only on $df(m)$, $dg(m)$, then one can define a Poisson tensor, i.e. a bilinear antisymmetric map $\pi_m:T_p^*M\times T_p^*M \to \R$ such that
\[ \pi_m (df(m),dg(m)) = \{f,g\}(m). \]
Note that the existence of a Poisson tensor implies almost localizability of the associated Poisson bracket: $\pb_U$ might turn out to not satisfy the Jacobi identity, as differentials of globally defined smooth functions at some point $m\in M$ might not span the whole $T^*_mM$. On the other hand, the differentials of the locally defined smooth functions span the whole $T^*_mM$. In effect, the Jacobi identity might not hold for all locally defined functions.

It was shown in \cite{BGT} that, even on the simplest Banach manifold given just as a Hilbert space, there exist localizable Poisson brackets, which depend on higher derivatives of functions (see also the construction of queer operational vector fields in \cite{michor}), hence cannot be given by a Poisson tensor. There is another situation, where Poisson brackets might not be defined by Poisson tensors. Namely, if one considers Poisson brackets on manifolds of class $C^k$, it is known that the space of derivations (corresponding to operational vector fields) is larger than the tangent bundle (corresponding to kinematic vector fields), see e.g. \cite{papy55,newns56}.
These observations motivated the authors of \cite{BGT} to reformulate Definition~\ref{def:OR} from \cite{OR} in the following way:

\begin{defn}[Belti\c t\u a--Goliński--Tumpach \cite{BGT}]\label{BGT}
A Banach Poisson manifold $(M,\pb)$ is a Banach manifold $M$ equipped with a localizable Poisson bracket $\pb$ for which there exists a Poisson tensor such that its corresponding sharp map $\sharp$ satisfies
\[ \sharp(T^*M)\subset TM.\]
\end{defn}

\begin{rem}
    A similar definition, using a Poisson anchor as a fundamental object, was given in \cite{pelletier12}. This approach was developed further in a more general setting, and we will present it in the next section.
\end{rem}
The definition presented in this section is suitable for a certain class of problems. It includes all Poisson brackets given by strong symplectic forms. Most examples not given by a strong symplectic form are Banach Lie--Poisson spaces, also introduced in \cite{OR}. For instance, there is a wide class of Poisson structures on predual spaces to $W^*$-algebras \cite{OR}. This approach allows to describe quantum mechanics in the framework of Hamiltonian mechanics in the spirit of \cite{bona}, but also it is useful e.g. to describe multidiagonal infinite Toda-like systems \cite{Oind}, or systems related to restricted Grassmannian \cite{GO-grass,GT-momentum}.

However, at the same time, this approach may be too restrictive and too rigid to include other interesting cases, such as weak symplectic manifolds or some Banach Poisson--Lie groups related to integrable systems (\cite{tumpach-bruhat}). For these reasons, in the next section, we will present a review of possible approaches to Poisson structures defined only on some subset of $\Ci(M)$ and discuss the differences between them.

\section{Poisson structures defined for subalgebras of smooth functions on a Banach manifold}\label{sec:weak}

In this section, we present three different methods for defining a Poisson structure on Banach manifolds that encompasses weak symplectic structures.
The first one presented in Section~\ref{sec:bundle_map} starts the construction by defining a bundle morphism $P$ from a subbundle $T^\flat M$ of the cotangent bundle $T^*M$ to $TM$. This bundle map is related to anchor maps in the theory of algebroids \cite{cabau-pelletier_book}. 
In the second method (Section~\ref{sec:bracket}), the notion of a Poisson algebra is central, and the construction starts with the definition of a Poisson bracket on an associative subalgebra $\mathcal{A}$ of the algebra $\mathcal{C}^\infty(M)$ of smooth functions on a Banach manifold $M$.
The last method presented in Section~\ref{sec:poisson_tensor}, is based on the notion of a Poisson tensor. A comparison of these different approaches will be given in Section~\ref{sec:comparison}.



\subsection{Poisson structures defined using bundle morphisms $P : T^\flat M \rightarrow  TM$}\label{sec:bundle_map}

In this section, we recall the definition of  \textbf{sub Banach Lie–Poisson manifolds} (sub P-manifolds for short) given in \cite{pelletier}, which was extended under the name of \textbf{partial Poisson manifolds} to convenient manifolds modeled on convenient spaces in \cite{pelletier19} and \cite[Chapter~7]{cabau-pelletier_book}. In particular, in the case of Banach manifolds, the notion of sub Banach Lie–Poisson manifolds is a particular example of partial Poisson manifolds (see Remark~\ref{sub_partial} below).  
In this approach, the main ingredient is a bundle morphism $P$ from a subbundle of the cotangent bundle $T^*M$ to the tangent bundle $TM$. 

\begin{defn}[Cabau--Pelletier \cite{pelletier}, Section 4.1]\label{sub_Poisson}
Let $q_M^\flat: T^\flat M \to M$ be a Banach subbundle of the cotangent bundle $q_M: T^* M \to M$. 
\begin{itemize}
\item A bundle morphism $P : T^\flat M \rightarrow  TM$ will be
called a \textbf{sub almost Poisson morphism}
(sub AP-morphism for short) if P is skew-symmetric relatively 
 to the duality pairing, i.e. \[\langle\alpha,P\beta\rangle = -\langle\beta, P\alpha\rangle\] for any $\alpha,\beta\in T^\flat M$. 

\item A sub almost Poisson morphism $P : T^\flat M \rightarrow  TM$ allows to define a $\R$-bilinear skew-symmetric pairing $\{ \cdot, \cdot\}_P$ on sections of $T^\flat M$ by 
\[
\{ \alpha, \beta\}_P = \langle \beta, P\alpha\rangle, \alpha, \beta\in\Gamma(T^\flat M).
\]
In particular, for any smooth function $f\in \Ci(M)$, one has
\begin{equation}\label{module}
\{ f\alpha, \beta\}_P = \langle \beta, f P\alpha\rangle = f \langle \beta, P\alpha\rangle.
\end{equation}

\item On the subalgebra \[\mathcal{F}^\flat = \{ f\in \Ci(M), df \in \Gamma(T^\flat M)\}\] of $\Ci(M)$, one can define a skew-symmetric bilinear map $\pb$,
which, by \eqref{module},  satisfies Leibniz rule:
\begin{equation}\label{def_bracket_via_morphism}
\{ f, g\} = \{df , dg\}_P = \langle df, P dg\rangle, 
\end{equation}
where $f, g\in \mathcal{F}^\flat$.
\end{itemize}
\end{defn}

It is not clear at first sight that the skew-symmetric bilinear map $\{\cdot, \cdot\}$ defined by \eqref{def_bracket_via_morphism} actually takes values in $\mathcal{F}^\flat$. However, this follows from the fact that $P$ takes values in $TM$ and from the fact that the inner product $df(X)$ of the differential of a function $f$ in $\mathcal{F}^\flat$ with any vector field $X$ on $M$ is again a function in $\mathcal{F}^\flat$. For a proof of this result, we refer to the following propositions.
\begin{prop}[Cabau--Pelletier \cite{pelletier19}, Lemma 2.1.1 or Cabau--Pelletier  \cite{cabau-pelletier_book}, Proposition 7.1]\label{prop:flat}
The bilinear map $\{\cdot, \cdot\}$ defined by \eqref{def_bracket_via_morphism} takes values in $\mathcal{F}^\flat$.
\end{prop}

\begin{defn}\label{subPoisson}
 When the bracket $\{\cdot, \cdot\}:  \mathcal{F}^\flat \times \mathcal{F}^\flat \rightarrow \mathcal{F}^\flat$ defined by \eqref{def_bracket_via_morphism} satisfies Jacobi identity, $\left(\mathcal{F}^\flat, \{\cdot, \cdot\}\right)$ has a Lie algebra structure, and we say that $\left(\mathcal{F}^\flat, \{\cdot, \cdot\}\right)$ is a \textbf{sub Banach Lie--Poisson manifold} (sub P-manifold for short).
\end{defn}








In order to compare previous definitions to their analogues in the more general setting of convenient manifolds \cite{michor}, below we recall the notion of partial Poisson manifolds introduced in \cite{pelletier19} and \cite{cabau-pelletier_book} for convenient manifolds.
\begin{defn}[Pelletier--Cabau \cite{pelletier19}, Section 2.1, \cite{cabau-pelletier_book}, chapter 7]\label{partial_Poisson} Let $M$ be a convenient manifold modelled on a convenient space $\mathbb{M}$ (cf. \cite{michor}). Let $p_M: TM \to M$ be it's kinematic tangent bundle and $p^*_M: T^*M \rightarrow M$ be it's kinematic cotangent bundle.
\begin{itemize}
    \item A vector subbundle $p': T^\flat M \to M$ of $p^*_M: T^*M \to M$ will be called weak subbundle if  $p'$ is a convenient bundle and the canonical injection $i:T^\flat M \to T^*M$ is a convenient bundle morphism. 
    \item Let $\mathcal{F}^\flat $ be the subalgebra of $\Ci(M)$ consisting of smooth functions $f: M\rightarrow \mathbb{R}$ for which  there exists $s\in \Gamma(T^\flat M)$ such that $i\circ s=df$. In other words, $\mathcal{F}^\flat $ is the subalgebra of smooth real functions on $M$ whose differentials are sections on $T^\flat M$.
    \item A convenient bundle morphism $P : T^\flat M \rightarrow  TM$ will be
called a \textbf{almost Poisson anchor}
if P is skew-symmetric relatively 
 to the duality pairing, i.e. \[\langle\alpha,P\beta\rangle = -\langle\beta, P\alpha\rangle\] for any $\alpha,\beta\in T^\flat M$. 

    \item For an almost Poisson anchor $P: T^\flat M \to TM$, we can define a almost-Lie bracket $\{ \cdot,\cdot \}: \mathcal{F}^\flat\times \mathcal{F}^\flat \to \mathcal{F}^\flat$ (i.e. a skew symmetric bilinear map satisfying Leibniz rule) on $\mathcal{F}^\flat$ by
    \begin{equation}\label{def_bracket_anchor}
\{ f, g\} = \langle df, P dg\rangle, f, g\in \mathcal{F}^\flat.
\end{equation}
\item We say that $\left(T^\flat M, M, P, \{\cdot, \cdot\} \right)$ is a \textbf{partial Poisson structure} on $M$ if the bracket $\{\cdot, \cdot\}$ defined by \eqref{def_bracket_anchor} satisfies the Jacobi identity
\[
\{f, \{g, h\}\} + \{g, \{h, f\}\} + \{h, \{f, g\}\} = 0,
\]
for any $f, g, h \in \mathcal{F}^\flat$. In this case $P$ is called a \textbf{Poisson anchor} and we say that $\left(M, \mathcal{F}^\flat, \{\cdot, \cdot\}\right)$ is a \textbf{partial Poisson manifold}.
    \end{itemize} 
\end{defn}

\begin{rem}\label{sub_partial}
    When we apply \textbf{Definition~\ref{partial_Poisson}} in the Banach setting, we observe that, unlike \textbf{Definition~\ref{sub_Poisson}}, the subbundle $T^\flat M$ is not necessarily assumed to be a Banach subbundle of $T^* M$. Hence \textbf{Definition~\ref{sub_Poisson}} is a particular case of \textbf{Definition~\ref{partial_Poisson}}.
\end{rem}

\begin{rem}
    In \cite{cabau-pelletier_book}, the subalgebra, on which the Poisson bracket is defined, is the set 
    $ \mathfrak{A}(U) $ of smooth functions $f\in \Ci(U)$ on an open set $U\subset M$ such that, for any $x\in U$,
    each iterated derivative $d^kf(x)$
    ($k\in \mathbb{N}^*$) satisfies:
    \[
    \forall (u_2, \dots, u_k)\in (T_xM)^{k-1}, d^k_x(\,\cdot\,, u_2, \dots, u_k)\in T_x^\flat M
    \]
    (see \cite[Definition 7.2]{cabau-pelletier_book}). However, it follows from \cite[48]{michor}, that $\mathfrak{A}(U)$ coincides with sections of $T^\flat M$ over $U\subset M$ (see also the proof of Lemma 2.1.1 in \cite{pelletier19}).
\end{rem}
\begin{rem}
As mentioned in \cite{pelletier} and \cite{cabau-pelletier_book} (see also \cite[Appendix B.1]{magri84} for the corresponding algebraic formulas), to an almost Poisson morphism $P : T^\flat M \rightarrow  TM$, one can associate an almost bracket $[\cdot, \cdot]_P$ on sections of $T^\flat M$ by the following formula (see \cite[Proposition~7.3(1)]{cabau-pelletier_book}):
\begin{equation}\label{bracket_one_forms}
    [\alpha, \beta]_P = \mathcal{L}_{P(\alpha)}\beta - \mathcal{L}_{P(\beta)}\alpha - d\langle \alpha, P(\beta)\rangle,
\end{equation}
where $\alpha, \beta \in \Gamma(T^\flat M)$. Note that identities $d(df) = d(dg) = 0$ and the Cartan formula for the Lie derivative $\mathcal{L}_X = \iota_X d + d \iota_X$ imply that
\[
    [df, dg]_P = d\{f, g\}
\]
for functions $f, g\in\mathcal{F}^\flat$. In particular, one can define a tensor of type $(2,1)$ denoted by $[\![P, P]\!]: \Gamma(T^\flat M)\times \Gamma(T^\flat M)\rightarrow \Gamma(T M)$ by
\begin{equation}\label{equ_7_in_CP12}
[\![P, P]\!](\alpha, \beta) = P\left([\alpha, \beta]_P\right) - [P\alpha, P\beta], \alpha, \beta \in \Gamma(T^\flat M),
\end{equation}
where $[\cdot, \cdot]_P$ is defined by \eqref{bracket_one_forms} and where $[\cdot, \cdot]$ denotes the bracket of vector fields. 
The fact that equation~\eqref{equ_7_in_CP12} indeed defines a tensor follows from computations in local coordinates (see Proposition~3.19, 1. page 123 in \cite{cabau-pelletier_book}).
In terms of the \textbf{Poisson tensor} $\pi: T^\flat M\times T^\flat M \rightarrow \mathbb{R}$ naturally associated to the almost Poisson morphism $P$ by 
\begin{equation}\label{pi_CP}
\pi(\alpha, \beta) = \langle \beta, P\alpha\rangle,
\end{equation}
one can consider the $(3, 0)$ tensor $[\![\![\pi, \pi]\!]\!]: T^\flat M \times T^\flat M\times T^\flat M \rightarrow \mathbb{R}$ defined by
\begin{equation}
    [\![\![\pi, \pi]\!]\!](\alpha, \beta, \gamma) = \langle \gamma, [\![P, P]\!](\beta, \alpha)\rangle
\end{equation}
also called the \textbf{Schouten bracket} of $\pi$ (see Section~7.2.2 in \cite{cabau-pelletier_book}).
According to Theorem~7.2 in \cite{cabau-pelletier_book}, the Jacobi identity for the bracket $\{\cdot, \cdot\}$ defined on $\mathcal{F}^\flat$ by \eqref{def_bracket_via_morphism} is equivalent to the fact that the tensor $ [\![\![\pi, \pi]\!]\!]$ vanishes identically. Remark that the identity $[\![\![\pi, \pi]\!]\!] \equiv 0$ a priori only implies that $[\![P, P]\!]$ takes values in
\[
\left(T^\flat M\right)^0 = \{ X \in TM \textrm{ such that } \langle \gamma, X\rangle = 0 \textrm{ for all  } \gamma\in T^\flat M\}.
\]
\end{rem}

\subsection{Poisson structures defined using Poisson algebras}\label{sec:bracket}
Here we present two different definitions of Poisson structures using the notion of Poisson algebra (see Definition~\ref{Poisson_algebra}). Their difference is subtle and will be discussed below.

\begin{defn}[Neeb--Sahlmann--Thiemann \cite{neeb14}, Definition 2.1]\label{NST}
Let $M$ be a smooth manifold modeled on a locally convex space. A \textbf{weak Poisson structure} on $M$ is a unital subalgebra $\mathcal A\subset \Ci(M,\R)$, i.e., it contains the constant functions and is closed under pointwise multiplication, with the following properties:
\begin{enumerate}
    \item[(P1)] $\mathcal A$ is endowed with a \textbf{Poisson bracket} $\pb$, this means that it is a Lie bracket, i.e.
    \[
\{f, g\} = -\{g, f\},\quad \{f, \{g, h\}\}+ \{g, \{h, f\}\}+ \{h, \{f, g\}\} = 0
    \]
    and it satisfies the Leibniz rule
    \[
\{f, gh\} = \{f, h\}h + g\{f, h\}.
    \]
    \item[(P2)] For every $m \in M$ and $v \in T_mM$ satisfying $dF(m)v = 0$ for every $F \in \mathcal A$, we have $v=0$. 
    \item[(P3)] For every $H \in \mathcal A$, there exists a smooth vector field $X_H \in \Gamma(TM)$ with $X_H F = \{F, H\}$ for $F\in\mathcal A$. It is called the Hamiltonian vector field associated with $H$.
\end{enumerate}
If (P1–P3) are satisfied, then we call the triple $(M, \mathcal A,\pb)$ a \textbf{weak Poisson manifold}.
\end{defn}

\begin{rem}
Note that by virtue of condition (P2), the Hamiltonian vector field associated with $H\in\mathcal A$ is unique.
\end{rem}

\begin{defn}[Belti\c t\u a--Odzijewicz \cite{beltita-odzijewicz}, Definition 2.1]\label{BO}
A Poisson structure on $M$ is a Lie algebra $(\mathcal P^\infty(M, \R), \pb)$, where $\mathcal P^\infty(M, \R)$ is an associative subalgebra of the algebra $\Ci(M, \R)$ of real smooth functions with fixed $\R$-linear map $\sharp : \mathcal P^\infty(M, \R) \to \Gamma(TM)$ such that:
\begin{enumerate}
    \item for any $f,g \in \mathcal P^\infty(M, \R)$ one has 
    \[
    (\sharp f)(g) = \{f, g\}.
    \]
    \item the algebra $\mathcal P^\infty(M, \R)$ separates vector fields $\xi\in \Gamma(TM)$ on $M$, i.e., if $\xi(f) = 0$ for any $f\in\mathcal P^\infty(M, \R)$ then $\xi \equiv 0$.
\end{enumerate}
\end{defn}

\begin{rem}
    Let us emphasize that the separability condition (P2) in Definition~\ref{NST} differs from the separability condition 2. in Definition~\ref{BO}. Indeed, in Definition~\ref{BO}, the Poisson algebra separates vector fields, whereas in Definition~\ref{NST} the differentials of functions in the Poisson algebra separate points in the tangent bundle. In Banach manifolds without bump functions, these two conditions may not coincide in general.
    
\end{rem}

   \subsection{Poisson structures defined using Poisson tensors}\label{sec:poisson_tensor}
In this section, the fundamental object is chosen to be a Poisson tensor. We give below two different definitions in this context. Note that \textbf{generalized Banach Poisson manifolds} were called \textbf{Banach sub-Poisson manifolds} in \cite[Definition 3]{GT-u-bial}.  


\begin{defn}[Tumpach, \cite{tumpach-bruhat}, Definition~3.4]\label{subbundleduality} 
We will say that $\mathbb{F}$ is a subbundle  of $T^*M$ \textbf{in duality} with the tangent bundle of $M$ if, for every $m\in M$, 
\begin{enumerate}
\item $\mathbb{F}_m$ is an injected Banach space of $T_m^*M$, i.e. $\mathbb{F}_m$ admits a Banach space structure such that the injection $\mathbb{F}_m\hookrightarrow T_m^*M$ is continuous,
\item the natural duality pairing between $T_m^*M$ and $T_mM$ restricts to a duality pairing between $\mathbb{F}_m$ and $T_mM$, i.e. $\mathbb{F}_m$ separates points in $T_mM$.
\end{enumerate}
\end{defn}

\begin{defn}[Tumpach \cite{tumpach-bruhat}, Definition 3.5, Definition 3.7]\label{T}
Let $M$ be a Banach manifold and $\mathbb{F}$ be a subbundle of $T^*M$ in duality with TM. Denote by $\Lambda^2 \mathbb F^*$ the vector bundle over M whose fiber over $m\in M$ is the Banach space of real continuous skew-symmetric bilinear maps on the subspace $\mathbb{F}_m$ of $T_m^*M$. 

\begin{enumerate}
    \item
A smooth section $\pi$ of $\Lambda^2 \mathbb F^*$ satisfies Jacobi identity iff:
\begin{enumerate}
    \item for any closed local sections $\alpha, \beta$ of $\mathbb F$, the differential $d(\pi(\alpha, \beta))$ is a local section of $\mathbb F$;
    \item for any closed local sections $\alpha$, $\beta$, one has
    $\gamma$ of $\mathbb{F}$,
    \begin{equation}\label{jacobi_tensor}
      \pi\left(\alpha, d\left(\pi(\beta, \gamma)\right)\right) + \pi\left(\beta, d\left(\pi(\gamma, \alpha)\right)\right) + \pi\left(\gamma, d\left(\pi(\alpha, \beta)\right) \right)= 0.
    \end{equation}
\end{enumerate}
In this case, it is called a \textbf{Poisson tensor} on $M$ with respect to $\mathbb{F}$.
\item A \textbf{generalized Banach Poisson manifold} is a triple $(M, \mathbb F, \pi)$ consisting of a smooth Banach manifold $M$, a subbundle $\mathbb F$ of the cotangent bundle $T^*M$ in
duality with $TM$, and a \textbf{Poisson tensor} $\pi$ on $M$ with respect to $\mathbb F$.
\end{enumerate}
\end{defn}

\begin{rem}
Given a Poisson tensor on a Banach manifold $M$ with respect to a subbundle $\mathbb{F}$, one defines a Poisson bracket on the space of locally defined functions with differentials in $\mathbb{F}$ by
\begin{equation}\label{bracket_tensor}
    \{ f, g\} = \pi(df, dg).
\end{equation}
Condition~1(a) in Definition~\ref{T} ensures that the bracket of two such functions is again a function with differentials in $\mathbb{F}$,
and condition~1(b) is equivalent to the usual Jacobi identity. Consequently, the space of smooth functions on $M$ with differentials in $\mathbb{F}$ forms a Poisson algebra. Note that the existence of Hamiltonian vector fields is not generally assumed. In Example~\ref{ex:u1} we give an example of Poisson structure in the sense of Definition~\ref{T}, where indeed not every function in the Poisson algebra possesses a Hamiltonian vector field. 
\end{rem}

\begin{rem}
In \cite[Lemma~5.8 (3)]{tumpach-bruhat} it was proved, in the context of Banach--Poisson Lie groups, that the left-hand side of equation~\eqref{jacobi_tensor} actually defines a tensor (the Schouten bracket of $\pi$). For the proof, right translations on the group were used. In this general context of Banach Poisson manifolds where the existence of Hamiltonian vector fields is not assumed, it is an open question to prove that the left-hand side of equation~\eqref{jacobi_tensor} defines a tensor.
\end{rem}

In \cite{li24}, a similar definition was introduced under the name of a projective Banach--Poisson manifold and used for the study of classical (projective) $r$-matrices and (projective) Yang--Baxter equations. We give it here for comparison despite the fact that the corresponding Poisson bracket is defined on the whole algebra of smooth functions. Let us first recall the definition of the projective tensor product.  To have a picture in mind, if one start with an Hilbert space $\mathcal{H}$ and its continuous dual $\mathcal{H}^*$, the space of real bounded bilinear maps on $\mathcal{H}\times\mathcal{H}^*$ can be identified with the space of bounded operators on $\mathcal{H}$, the injective tensor product of $\mathcal{H}^*$ and $\mathcal{H}$ can be identified with the Banach space of compact operators on $\mathcal{H}$, whereas the projective tensor product of $\mathcal{H}^*$ and $\mathcal{H}$ can be identified with the Banach space of trace class operators on $\mathcal{H}$.
In fact, the projective tensor product is the Banach space that solves the universal property of tensor products in the category of Banach spaces, see \cite[Theorem 2.9]{ryan}. This means the following:
\begin{prop}[{Universal property for Banach tensor product}] \label{prop:tensor}
Given two Banach spaces $\mathfrak{g}_1$ and $\mathfrak{g}_2$ over the same field $\mathbb{K}$, there is a unique (up to isomorphism) Banach space $\mathfrak{g}_1\oth\mathfrak{g}_2$ such that there exists a continuous bilinear map $\iota:\mathfrak{g}_1\times\mathfrak{g}_2\rightarrow \mathfrak{g}_1\oth\mathfrak{g}_2$ with the following universal property:
if $A : \mathfrak{g}_1\times\mathfrak{g}_2\rightarrow \mathfrak{g}$ is any continuous bilinear mapping into a Banach space $\mathfrak{g}$, then there exists a unique continuous linear mapping $\hat A : \mathfrak{g}_1\oth\mathfrak{g}_2 \rightarrow \mathfrak{g}$  such that $A = \hat A \circ \iota$ and $\|\hat A\| = \|A\|$.
\[
\xymatrix{
	\mathfrak g_1 \times \mathfrak g_2 \ar^{\iota}[r] \ar_A[dr] & \mathfrak g_1 \oth \mathfrak g_2 \ar@{-->}^{\hat A}[d]\\
	& \mathfrak g
	}
\]
\end{prop}
 

In particular, one has a Banach isomorphism between the continuous dual of the projective tensor product $\mathfrak{g}_1\oth\mathfrak{g}_2$ and the space $B(\mathfrak{g}_1, \mathfrak{g}_2; \mathbb{K})$ of $\mathbb{K}$-valued continuous bilinear maps on $\mathfrak{g}_1\times\mathfrak{g}_2$, which itself is isomorphic (as a Banach space) to the space of linear maps from $\mathfrak{g}_1$ into $\mathfrak{g}_2^*$:
\begin{equation}\label{isom}
(\mathfrak{g}_1\oth\mathfrak{g}_2)^* \cong B(\mathfrak{g}_1, \mathfrak{g}_2; \mathbb{K}) \cong L(\mathfrak{g}_1; \mathfrak{g}_2^*).
\end{equation}

\begin{defn}[Li--Wang \cite{li24}, Definition 4.1, Definition 4.3]\label{LW}
Let $M$ be a Banach manifold and $\bar{\Lambda}^2 TM$ be the vector bundle over $M$, whose fiber over $m$ is in the space of skew-symmetric elements in the tensor product $T_{m}M\oth T_{m}M $, denote by $\Lambda^2 TM$ the vector bundle over $M$ whose fiber over $m$ is the space of skew-symmetric continuous bilinear forms on $T_{m}M$, and and denote by $\Lambda^2 T^*M$ the vector bundle over $M$ whose fiber over $m$ is the space of skew-symmetric continuous bilinear forms on $T_{m}^*M$.
\begin{itemize}
\item A smooth section $\pi$ of $\bar{\Lambda}^2 TM$ is called a \textbf{projective Poisson bivector} if for any closed local sections $\alpha, \beta, \gamma \in T^*M$, 
\begin{equation}\label{poisson_li}
      \left(\alpha\oth d\left((\beta\oth \gamma)(\pi)\right)\right) (\pi)+ \left(\beta\oth d\left((\gamma\oth \alpha)(\pi)\right)\right) (\pi)+ \left(\gamma\oth d\left((\alpha, \beta)(\pi)\right) \right)(\pi)= 0,
    \end{equation}    
where $\beta\oth \gamma$ lies in $\Lambda^2TM$ such that $(\beta\oth\gamma)_m(\pi_m) = (\beta_m\oth \gamma_m)(\pi_m)$ and $d\left((\beta\oth \gamma)(\pi)\right)$ is the differential of $(\beta\oth \gamma)(\pi)$.
    
    \item The pair $(M,\pi)$ is called the \textbf{ projective Banach--Poisson manifold} if $\pi$ is a projective Poisson bivector on $M$.
    \end{itemize}
    \end{defn}


\begin{rem}
For instance, as mentioned above, for a Hilbert manifold modeled on a Hilbert space $\mathcal{H}$, a projective Poisson tensor may be identified with a trace class operator, whereas a Poisson tensor in the sense of Definition~\ref{T} will correspond to an operator which is only supposed to be bounded. 
\end{rem}


Definition~\ref{LW} allows to define a Poisson bracket on the whole space of smooth functions on $M$ by 
\begin{equation}\label{pli}
\{f, g\} = \left(df \oth dg\right)(\pi),
\end{equation}
where $f, g \in \Ci(M)$.

\begin{rem}\label{rem:strong_jacobi}
Note that, as was mentioned in Remark~3.2(2) in \cite{tumpach-bruhat} and Remark~4.2(a) in \cite{li24}, on Banach manifolds without bump functions, there may be elements of the cotangent space $T_m^*M$ or of the subspace $\mathbb{F}_m$ that are not the differential at $m$ of a globally defined smooth map $f:M \rightarrow \mathbb{R}$.
This means that differentials of globally defined smooth functions may only generate a proper subspace of $\mathbb F$ or $T^*M$ in general. Hence identities~\eqref{jacobi_tensor} and \eqref{poisson_li} imply that the brackets $\pb$ defined by equation~\eqref{bracket_tensor} and \eqref{pli} satisfy Jacobi identity, but the converse is not necessarily true.
\end{rem}

\section{Comparison of different approaches to Poisson Geometry in infinite-dimensions}\label{sec:comparison}
\subsection{Poisson algebra of functions vs subbundles of the cotangent bundle}\label{sec:bundle-algebra}

One of the differences between the definitions presented in the previous section is the choice of the object used to distinguish a class of functions, on which the Poisson bracket will be defined. Definition~\ref{partial_Poisson} and Definition~\ref{T} define a certain subbundle $\mathbb F$ of the cotangent bundle $T^*M$. The Poisson bracket is then defined on the functions with the differential being a section of $\mathbb F$. It leads to an associative algebra $\mathcal A^{\mathbb F}\subset \Ci(M)$ defined as 
\[ \mathcal A^{\mathbb F} = \{ f\in \Ci(M) \;|\; df\in \Gamma(\mathbb F)\}. \]
On the other hand, Definition~\ref{NST} and Definition~\ref{BO} start with a choice of an algebra. 
Having an algebra $\mathcal A$, one can introduce a co-distribution $\mathbb F^{\mathcal A}\subset T^*M$ spanned at each point by differentials of elements in $\mathcal A$ (it is called the co-characteristic distribution in \cite{beltita-odzijewicz} and the first jet $J^1(\mathcal{A}$) of $\mathcal{A}$ in \cite{tumpach-bruhat}):
\[
\mathbb F^{\mathcal A}_m =\{df_m \;|\; f \in\mathcal{A}\} \subset T_m^*M.
 \]
This co-distribution yields another algebra $\widetilde{\mathcal A}=\mathcal A^{(\mathbb F^{\mathcal A})}$ defined, like before, as the collection of functions with differential being a section of the co-distribution $\mathbb F^{\mathcal A}$. In general $\widetilde{\mathcal A}$ might be larger than $\mathcal A$. For one thing $\widetilde{\mathcal A}$ is always unital. For example, if one takes as $\mathcal A$ the algebra of all even functions $f$ on $\mathbb R\setminus\{0\}$, i.e. satisfying $f(-x) = f(x)$, one obtains $\widetilde{\mathcal A}=\Ci(\mathbb R\setminus\{0\})$ since the dependence of the derivatives at $x$ and $-x$ is lost in the process. Note that in general $\mathbb F^{\mathcal A}$ may not be a vector bundle. Indeed, by replacing $\mathbb R\setminus\{0\}$ by $\mathbb R$ in the previous example, we see that the co-distribution associated with the algebra of even functions on $\mathbb{R}$ is equal to the whole cotangent space on $\mathbb{R}\setminus \{0\}$, but reduces to $\{0\}$ at $0\in \mathbb{R}$.

Since both Definition~\ref{NST} (when restricted to Banach setting) and Definition~\ref{BO} assume the existence of Hamiltonian vector fields, they imply that Poisson brackets depend only on the first differentials. Thus, it implies the existence of a kind of ``Poisson tensor'' $\pi:\mathbb F^{\mathcal A}\to TM$, satisfying
\begin{equation}\label{poisson_tensor}
    \pi(df,dg) = \{f,g\}
\end{equation}
for $f,g\in\mathcal A$. It is then a sufficient object to extend the Poisson bracket to $\widetilde{\mathcal A}$ via the same formula \eqref{poisson_tensor} for $f,g\in\widetilde A$. However in general $\pi$ may not be a smooth section of any bundle, only a collection of bilinear antisymmetric maps defined on fibers of co-distribution $\mathbb F^{\mathcal A}$. That leads to numerous problems. First of all, there is no reason to expect $\{f,g\}$ to be smooth or to belong to $\widetilde{\mathcal A}$, see \cite[Remark 2.3]{neeb14}. Moreover, it may not satisfy Jacobi identity or have smooth Hamiltonian vector fields. One way around this difficulty is  to consider the subalgebra \[\mathcal{B}= \{f\in \widetilde{\mathcal{A} }\;|\; \pi(df,\cdot) = X_f \in \Gamma(TM)\}.\]
For $f,g\in \mathcal{B}$, then the bracket  $\{f,g\}=X_f(g)=dg(X_f)$ is in $\Ci(M)$ since $X_f$ and $dg$ are smooth sections of $TM$ and $T^*M$ respectively.

On the other hand, if we start with a bundle (or co-distribution) $\mathbb F$, and consider the algebra of globally defined functions $\mathcal A^{\mathbb F}$, it may happen in general that the differentials of functions from $\mathcal A^{\mathbb F}$ do not span the whole $\mathbb F$. It is a consequence of the lack of bump functions on $M$, see also Remark~\ref{rem:strong_jacobi}. Therefore it may happen that $\mathbb F^{(\mathcal{A}^{\mathbb F})}$ might be strictly contained in $\mathbb F$. However an example of such a situation is not yet known.

Note that it is often interesting to study canonical Poisson structures on a Cartesian product of two Banach manifolds $M_1\times M_2$. If we deal with algebras of functions $\mathcal A_1\subset \Ci(M_1)$, $\mathcal A_2\subset \Ci(M_2)$, one needs to consider the algebra of functions on the product Banach manifold $M_1\times M_2$ which project to functions in $\mathcal A_1$ (resp. $\mathcal{A}_2$) via the projection on the first (resp. second) factor (in the spirit of \cite[Theorem 2.2]{OR}). On the other hand, working with subbundles $\mathbb F_1\subset T^*M_1$ and $\mathbb F_2\subset T^*M_2$, it is straightforward to define the Whitney sum of them $\mathbb F_1\oplus \mathbb F_2$, see \cite[Proposition 5.1]{tumpach-bruhat}.

\subsection{Localizability}

The notion of localizability for Poisson brackets on $\Ci(M)$ was given in Definition~\ref{def:localizable}, see \cite{BGT}. This property is important e.g. in the situation when one attempts to do some calculations in a chart. Moreover, in general, Hamiltonian vector fields might not be complete, thus local point of view on the dynamical system might be more convenient.

For Poisson brackets on some subalgebra of $\Ci(M)$ the situation is more complicated.

For the brackets defined via Poisson tensor (Definition~\ref{T}), one can easily define the localization for Poisson bracket as follows:
\[ \pb_U : \mathcal A_U^{\mathbb F} \times \mathcal A_U^{\mathbb F} \to \mathcal A_U^{\mathbb F}\]
where $U\subset M$ is open and
\[ \mathcal A_U^{\mathbb F} = \{ f \in \Ci(U)\;|\; df\in \Gamma(\mathbb F_{|U})\}. \]
Conditions 1(a) and 1(b) from Definition~\ref{T} guarantee that $\pb_U$ takes value in the same algebra $\mathcal A_U^{\mathbb F}$ and satisfies Jacobi bracket, see also Remark~\ref{rem:strong_jacobi}.

Now, in the case of Poisson brackets defined in Definition~\ref{NST} and Definition~\ref{BO}, the discussion in the previous subsection applies, i.e. one might define
\[ \widetilde{\mathcal A}_U = \{ f \in \Ci(U)\;|\; df\in \Gamma(\mathbb F^{\mathcal A}_{|U})\}. \]
However in this case, there is no reason to expect in general to get a Poisson algebra structure on $\widetilde{\mathcal A}_U$: neither smoothness of Poisson bracket nor Jacobi are guaranteed.

The Poisson brackets defined in Definition~\ref{partial_Poisson}, on the first glance, seem similar to the ones defined in Definition~\ref{T}. Since bundle morphism $P$ is by definition smooth, the localized brackets on $\mathcal A_U^{\mathbb F}$ still produce smooth functions and Proposition~\ref{prop:flat} can be generalized to conclude that the result is still in the same local algebra of functions $\mathcal A_U^{\mathbb F}$. However, Jacobi identity was assumed only for globally defined functions, and as result, it might not hold for locally defined ones, see Remark~\ref{rem:strong_jacobi}.


\subsection{Examples of Poisson brackets}

In this section, we emphasize through examples the differences between the definitions of Poisson structures given in section~\ref{sec:weak}. 


\begin{example}
Let us begin by building on the example in Section~\ref{sec:bundle-algebra}. Consider $M=\R^2$ with the canonical Poisson bracket $\{f,g\} = \frac{\partial f}{\partial x}\frac{\partial g}{\partial y}-\frac{\partial f}{\partial y}\frac{\partial g}{\partial x}$. Define a subalgebra $\mathcal A$ of $\Ci(\R^2)$ consisting of functions even with respect to the first variable, i.e. $f(x,y)=f(-x,y)$ for all $x,y\in\R$. The restriction of the canonical Poisson bracket to $\mathcal A$ is a Poisson bracket in the sense of Definitions \ref{NST} and \ref{BO}, but not with respect to the others, as the corresponding co-characteristic distribution $\mathbb F^{\mathcal A}$ does not have a constant rank and is not a subbundle of $T^*M$.
\end{example}

\begin{example}
A weak symplectic structure on a Banach manifold allows one to define a Poisson bracket in the sense of Definition~\ref{partial_Poisson} (\cite[Example 2.2.4]{pelletier19}) and Definition~\ref{T} (\cite[Proposition 3.11]{tumpach-bruhat}).

In order for this Poisson bracket to satisfy Definition~\ref{sub_Poisson} one needs to make sure that the range of the flat map $\flat$ defined by the weak symplectic form constitutes a Banach subbundle, see \cite[Example 4.2]{pelletier}.
 
In order for this Poisson bracket to satisfy Definition~\ref{NST} or Definition~\ref{BO}, one needs to impose an additional condition on the weak symplectic form to ensure that the algebra of functions on which the Poisson bracket is defined has the required separation property, see
\cite[Proposition~2.18]{neeb14} and \cite[Proposition~2.2.2]{beltita-odzijewicz}.
\end{example}

\begin{example}\label{ex:u1}
As manifold $M$ consider the Lie algebra $\mathfrak{u}_1(\mathcal{H}) = M$ consisting of skew-symmetric trace class operators on an infinite-dimensional separable Hilbert space $\mathcal{H}$. Consider a decomposition of  $\mathcal{H}$ into the sum of two infinite-dimensional orthogonal closed subspaces $\mathcal{H} = \mathcal{H}_+\oplus \mathcal{H}_-$. We will denote by $p_+$ the orthogonal projection onto $\mathcal{H}_+$. The dual space of $\mathfrak{u}_1(\mathcal{H})$ can be identified via the trace $\Tr$ with the space $\mathfrak{u}(\mathcal{H})$ consisting in skew-symmetric bounded operators, see e.g. \cite[Proposition 2.1]{Ratiu-grass}. Hence the differential at $\mathfrak{a}\in \mathfrak{u}_1(\mathcal{H})$ of a smooth function $f$ defined on $\mathfrak{u}_1(\mathcal{H})$ is an element in $\mathfrak{u}(\mathcal{H})$. Consider the subbundle $\mathbb{F}$ of $T^*\mathfrak{u}_1(\mathcal{H})$, whose fiber over $\mathfrak{a}\in \mathfrak{u}_1(\mathcal{H})$ is the space $\mathfrak{u}_2(\mathcal{H})$ consisting of skew-symmetric Hilbert--Schmidt operators:
    \[
\mathbb{F}_\mathfrak{a} = \mathfrak{u}_2(\mathcal{H}) \subset \mathfrak{u}(\mathcal{H}) \simeq \mathfrak{u}_1^*(\mathcal{H}) \simeq T_\mathfrak{a}^* \mathfrak{u}_1(\mathcal{H}),
    \]
for any $\mathfrak{a} \in \mathfrak{u}_1(\mathcal{H})$. Define the following antisymmetric bilinear form $\pi\in \Lambda^2\mathbb{F}^*$ on $\mathbb{F}$ by 
    \[
\pi_\mathfrak{a}(A, B)  = \operatorname{Tr}(i p_+ [A, B]),
    \]
where $A, B\in \mathfrak{u}_2(\mathcal{H}) = \mathbb{F}_\mathfrak{a}$ and  $[A, B] = AB - BA$.
Note that since $p_+ B$ and $A$ are Hilbert--Schmidt, we have
\[
    \pi_\mathfrak{a}(A, B)  = \operatorname{Tr}(i p_+ AB - i p_+ BA) = i \operatorname{Tr}( p_+ AB - A p_+ B) = \operatorname{Tr}([i p_+, A] B).
\]
Since $\pi$ is a constant bilinear form on $M = \mathfrak{u}_1(\mathcal{H})$, it satisfies Jacobi identity \eqref{jacobi_tensor} (see also the proof of \cite[Theorem~3.14]{tumpach-bruhat}, where the Jacobi identity is proved for Lie--Poisson brackets).

Denote by $\mathcal{A}^\mathbb{F}$ the subalgebra of $\Ci(M)$ consisting of those functions whose differential at any point $\mathfrak{a}\in \mathfrak{u}_1(\mathcal{H})$ can be identified with an element in  $\mathfrak{u}_2(\mathcal{H}) = \mathbb{F}_\mathfrak{a}.$
In particular, linear functions  $f:\mathfrak{u}_1(\mathcal{H})\rightarrow \mathbb{R}$ of the form \begin{equation}\label{linear} f(\mathfrak{b}) = \operatorname{Tr}(A \mathfrak{b}), \end{equation} with $A$ a skew-hermitian Hilbert--Schmidt operator, belong to $\mathcal{A}^\mathbb{F}$. For such a function, one has $d_\mathfrak{a}f(\mathfrak{b})  = \operatorname{Tr} (A \mathfrak{b})$ for all $\mathfrak{a}\in \mathfrak{u}_1(\mathcal{H})$, hence $d_\mathfrak{a} f$  can be identified with  $A$ via the trace. It follows that 
    \[
\pi_\mathfrak{a}(d_\mathfrak{a}f, B) = \operatorname{Tr}([i p_+, d_\mathfrak{a}f] B) = \operatorname{Tr}([i p_+, A] B),\]
for any $\mathfrak{a}\in M = \mathfrak{u}_1(\mathcal{H})$ and any $B\in \mathbb{F}_\mathfrak{a} = \mathfrak{u}_2(\mathcal{H})$. In particular, the derivation of the algebra $\mathcal{A}^\mathbb{F}$ associated to the function $f$ via $\pi$ is given by
\[
\operatorname{Der}_f(g) =  \operatorname{Tr}([i p_+, A] dg),
\]
where $g\in \mathcal{A}^\mathbb{F}$. Consider the block decomposition of $A$ with respect to the direct sum $\mathcal{H} = \mathcal{H}_+\oplus \mathcal{H}_-$
\[
A = \left(\begin{array}{cc} A_{++} & A_{+-}\\ -A_{+-}^* & A_{--}\end{array}\right),
\]
then 
\[
[i p_+, A] = \left(\begin{array}{cc} 0 & i A_{+-}\\ -i A_{+-}^* & 0 \end{array}\right).
\]
It follows that the derivation $\operatorname{Der}_f$ corresponds to a vector field on $M = \mathfrak{u}_1(\mathcal{H})$ if and only if $A_{+-}$ is a trace class operator. More generally, any function $f$ in $\mathcal{A}^\mathbb{F}$ whose differential at any point $\mathfrak{a}\in\mathfrak{u}_1(\mathcal{H})$ in block decomposition has co-diagonal blocks Hilbert--Schmidt but not trace class will not admit a Hamiltonian vector field. In consequence, this example is a Poisson manifold in the sense of Definition~\ref{T}, but not for any other definition presented in this paper.
\end{example}
\begin{example}
    In \cite[Theorem~14]{GT-u-bial}, the authors introduced a Poisson structure in the sense of Definition~\ref{T} on the unitary group $U(\mathcal{H})$ of an infinite-dimensional separable Hilbert space, with subbundle $\mathbb{F}\subset T^*U(\mathcal{H})$ defined by the right translations of the subspace $L_1(\mathcal{H})/\mathfrak{u}_1(\mathcal{H})\subset \mathfrak{u}^*(\mathcal{H})$ (see \cite[equation~(6)]{GT-u-bial}). In particular, $\mathbb{F}$ is not a Banach subbundle of $T^*U(\mathcal{H})$. Hence, this example is not a Poisson manifold in the sense of Definition~\ref{sub_Poisson}.
\end{example}

\subsection{Summary of the differences between previous approaches}

Let us list here the different properties that are useful in the study of Poisson manifolds, and point out the ones which are implied by the definitions we considered here.
\begin{enumerate}
    \item \textbf{Jacobi identity for a subalgebra $\mathcal{A}$ of $\Ci(M)$} (see Definition~\ref{def:Jacobi}). In \cite{pelletier} almost Poisson brackets are also considered, but we did not focus on that case. Thus, all definitions satisfy the Jacobi identity for globally defined smooth functions on $M$. 

    \item \textbf{Jacobi identity for an antisymmetric tensor $\pi$ on a subbundle $\mathbb{F}\subset T^*M$} (see Definition~\ref{T}~1.): it is assumed only in Definitions~\ref{T}~1 and \ref{LW} to hold on the level of Poisson tensor for locally defined closed sections. Due to the lack of bump functions in general, it might not hold in all cases, see Remark~\ref{rem:strong_jacobi} and \cite[Remark 2.3]{neeb14}.
    \item \textbf{Localizability of Poisson bracket} (see Definition~\ref{def:localizable}): all definitions imply the existence of a local almost Poisson bracket due to the existence of some kind of Poisson tensor or Poisson anchor. However, these local Poisson brackets might not satisfy Jacobi identity or even preserve smoothness in general.
    
    \item \textbf{Leibniz identity for a subalgebra $\mathcal{A}$ of $\Ci(M)$} (see Definition~\ref{def:Leibniz}): note that Leibniz identity does not imply the existence of Poisson tensor by itself, see \cite{BGT}.
    
    \item \textbf{Separability of globally defined vector fields} by a subalgebra $\mathcal{A}$ of $\Ci(M)$ (see Definition~\ref{BO}~2.)
    \item \textbf{Separability of tangent vectors} by a subalgebra $\mathcal{A}$  of $\Ci(M)$ (see Definition~\ref{NST}~(P2)): for Banach manifold it is not automatic that every tangent vector is a value of some globally defined vector field, thus this condition may be stronger than separability of globally defined vector fields
    \item \textbf{Vector bundle structure} on $\mathbb{F}^\mathcal{A}$ (see section~\ref{sec:bundle-algebra}): Definitions~\ref{NST} and \ref{BO} do not assume existence of vector bundle structure. In Definition~\ref{sub_Poisson} the bundle is required to be a Banach subbundle, while in Definition~\ref{T} it is only assumed to be subbundle with its own Banach structure.
    \item \textbf{Existence of Poisson anchor} (see equation~\ref{def_bracket_anchor}): a linear map from $\mathbb{F}^\mathcal{A}$ to $TM$. With Definition~\ref{T}, it takes values in $T^{**}M$ in general.
    \item \textbf{Existence of Poisson tensor} (see equation~\ref{poisson_tensor}): requires the existence of bundle structure on $\mathbb{F}^\mathcal{A}$ and it should be a smooth section of $\bigwedge^2(\mathbb{F}^\mathcal{A}$

    \item \textbf{Existence of Hamiltonian vector fields} for all functions on which the Poisson bracket is defined (see e.g. Definition~\ref{NST}~(P3)): satisfied by all definitions except Definition~\ref{T}. In \cite[Remark 4.5]{li24} the existence is claimed without proof.
\end{enumerate}

\newcommand{\chk}{\textcolor{LimeGreen}{$\surd$}}
\newcommand{\no}{}
\newcounter{rownumbers}
\newcommand\rownumber{\stepcounter{rownumbers}\arabic{rownumbers}. }
\begin{table}[!ht]
\begin{tabular}{|l|w{c}{3em}|w{c}{3em}|w{c}{3em}|w{c}{3em}|w{c}{3em}|w{c}{3em}|}
\hline 
Reference & \cite{pelletier} & \cite{cabau-pelletier_book} & \cite{neeb14} & \cite{beltita-odzijewicz} & \cite{tumpach-bruhat} & \cite{li24}\\
\hline
  Definition number  & \ref{sub_Poisson} & \ref{partial_Poisson} & \ref{NST} &  \ref{BO}   & \ref{T} & \ref{LW}\\
     \hline
     \textbf{\rownumber Jacobi for $\mathcal{A}$ } & \chk&\chk&\chk&\chk&\chk&\chk\\
     \hline
     \textbf{\rownumber Jacobi for $\pi$ on $\mathbb{F}^\mathcal{A}$} &\no&\no&\no&\no&\chk&\chk\\
     \hline
     \textbf{\rownumber Localizability} &\no&\no&\no&\no&\chk&\chk\\
     \hline
     \textbf{\rownumber Leibniz for $\mathcal{A}$} &
     \chk&\chk&\chk&\chk&\chk&\chk\\\hline
     \textbf{\rownumber Separability of $\Gamma(TM)$} &\no&\no&\chk&\chk&\chk&\chk\\
     \hline
     \textbf{\rownumber Separability of $TM$} &\no&\no&\chk&\no&\chk&\chk\\
     \hline
     \textbf{\rownumber $\mathbb{F}^\mathcal{A}$ vector bundle}   &\chk&\chk&\no&\no&\chk&\chk\\
     \hline
     \textbf{\rownumber Poisson anchor} &\chk&\chk&\chk&\chk&\no&\no\\
     \hline
     \textbf{\rownumber Poisson tensor} &\chk&\chk&\no&\no&\chk&\chk\\
     \hline
     \textbf{\rownumber Hamiltonian v.f.} &\chk&\chk&\chk&\chk&\no&\\
     \hline
\end{tabular}
\caption{Summary of the properties of the different definitions of Poisson structures presented in Section~\ref{sec:weak}. 
A \chk{} sign means that the property is implied by the definition. 
Otherwise,
the property is 
not necessarily satisfied.}
\end{table}

Let us also add that the queer Poisson brackets discussed in Section~\ref{sec:strong} and in \cite{BGT}, which depend on higher derivatives of functions, have all listed properties except Poisson anchor, Poisson tensor and Hamiltonian vector fields.

From the observations made in this section, one can conjecture that the most general approach to Poisson brackets would be to define it using a co-distribution on a Banach manifold satisfying certain conditions. In this manner, all definitions stated here would be included as special cases. However one should take care to impose conditions strong enough to guarantee the required properties, e.g. localizability.

\end{document}